\documentclass[journal, draftclsnofoot,oneside, onecolumn,11pt, letterpaper]{IEEEtran}

\ifCLASSINFOpdf

\else

\fi

\usepackage{pslatex}
\usepackage{amsfonts,color,morefloats}
\usepackage{amssymb,amsmath,latexsym,amsthm}
\usepackage{bbm}
\usepackage{amsfonts}
\usepackage{mathrsfs}
\usepackage{amsthm}
\usepackage{latexsym}
\usepackage{color}
\usepackage{dcolumn}
\usepackage{multicol}
\usepackage{extarrows}
\usepackage{amsmath}
\usepackage{cite}
\usepackage{amsmath,bm}
\allowdisplaybreaks[4]

\newtheorem{theorem}{Theorem}
\newtheorem{lemma}[theorem]{Lemma}

\newtheorem{definition}{Definition}

\newtheorem*{conjecture}{Conjecture}
\newtheorem*{remark}{Remark}

\makeatletter

\newcommand{\Rmnum}[1]{\expandafter\@slowromancap\romannumeral #1@}
\makeatother

\begin{document}
	\title{Arithmetic autocorrelation distribution of binary $m$-sequences
	}

	\author{Xiaoyan~Jing,
		Aixian~Zhang
		and Keqin~Feng
		\thanks{Xiaoyan Jing is with Research Center for Number Theory and Its Applications, Northwest University, Xi'an 710127, China (e-mail: jxymg@126.com).}
		\thanks{ Aixian Zhang is with Department of Mathematical Sciences, Xi'an University of Technology, Xi'an 710054, China (e-mail: zhangaixian1008@126.com).}
		\thanks{ Keqin~Feng is with Department of Mathematical Sciences, Tsinghua University, Beijing 100084, China (e-mail: fengkq@tsinghua.edu.cn).}
		\thanks{This work is supported by the National Natural Science Foundation of China under Grant No. 12031011 and  the Natural Science Basic Research Plan in Shaanxi Province of China under Grant No. 2022JM-017.
		}
	}
	
	%
	%

\markboth{}%
{Shell \MakeLowercase{\textit{et al.}}: Bare Demo of IEEEtran.cls for IEEE Journals}
%

\maketitle

\IEEEpeerreviewmaketitle

\begin{abstract}
	Binary $m$-sequences are ones with the largest period $n=2^m-1$ among the binary sequences produced by linear shift registers with length $m$.  They have a wide range of applications in communication since they have several desirable pseudorandomness such as balance, uniform pattern distribution and ideal (classical) autocorrelation. In his reseach on arithmetic codes, Mandelbaum \cite{9Mand} introduces a 2-adic version of classical autocorrelation of binary sequences, called arithmetic autocorrelation. Later, Goresky and Klapper \cite{3G1,4G2,5G3,6G4} generalize this notion to nonbinary case and develop several properties of arithmetic autocorrelation related to linear shift registers with carry. Recently, Z. Chen et al. \cite{1C1} show an upper bound on arithmetic autocorrelation of binary $m$-sequences and raise a conjecture on absolute value distribution on arithmetic autocorrelation of binary $m$-sequences.
	
	In this paper we present a general formula for computing arithmetic autocorrelations, from which we totally determine the arithmetic autocorrelation distribution of arbitrary binary $m$-sequences. Particularly, the conjecture raised in \cite{1C1} is verified.
\end{abstract}


\IEEEpeerreviewmaketitle

\section{Introduction}\label{sec1}
Let $q=2^m\ (m\geq2)$, $\mathbb{F}_q$ be the finite field with $q$ elements, $T: \mathbb{F}_q\rightarrow\mathbb{F}_2$ be the trace mapping defined by
$$T(a)=a+a^2+a^{2^2}+\cdots+a^{2^{m-1}}\quad (a\in\mathbb{F}_q).$$

Let $\pi$ be a primitive element of $\mathbb{F}_q$, $\mathbb{F}^{\ast}_q=\mathbb{F}_q\setminus\{0\}=\langle\pi\rangle.$ The binary $m$-sequence with period $n=2^m-1$ is defined by
$$s=(s_{\lambda}=T(\pi^{\lambda}))^{n-1}_{\lambda=0}$$
and shift sequences $s^{(\tau)}=(s^{(\tau)}_{\lambda})^{n-1}_{\lambda=0}\ (1\leq\tau\leq n-1)$ are defined by $s_{\lambda}^{(\tau)}=s_{\lambda+\tau}=T(\pi^{\lambda+\tau})\ (s_{\lambda}^{(0)}=s_{\lambda})$.

Binary $m$-sequences have a wide range of applications in the field of communications such as error-correcting coding, spread spectrum communications, code divison multiple access systems and cryptography since such sequences have the longest period $n=2^m-1$  among all binary sequences generated by linear shift registers with length $m$ and have good pseudorandomness. We list several basic properties of binary $m$-sequences which we need in this paper.

\begin{lemma}\label{le1}
	Let $s=(s_{\lambda})^{n-1}_{\lambda=0}$ be a binary $m$-sequence with period $n=2^m-1$. Then
	
	(1).\ (linearity) All shift sequences $s^{(\tau)}\ (0\leq\tau\leq n-1)$ plus zero sequence form a $\mathbb{F}_2$-vector space of dimension $m$.
	
	(2).\ (uniform pattern distribution) For $1\leq l\leq m$, $a=(a_1,\ldots,a_l)\in \mathbb{F}_2^l$, let
	$$N(a)=\sharp\{0\leq i\leq n-1\mid (s_i,s_{i+1},\ldots,s_{i+l-1})=a\}\quad(\text{where}\  s_{n+\lambda}=s_{\lambda} ).$$
	Then
	\begin{align*}
		N(a)=\left\{ \begin{array}{ll}
			2^{m-l}-1,                             & \mbox{ if}\ a_1=a_2=\cdots=a_l=0 \\
			2^{m-l}, & \mbox{otherwise}
		\end{array}
		\right.
	\end{align*}
	(3).\ (ideal classical autocorrelation) $\sum\limits_{\lambda=0}^{n-1}(-1)^{s_{\lambda
		}+s_{\lambda+\tau}}=-1$ for all $\tau$, $1\leq\tau\leq n-1.$
\end{lemma}

In 1967, Mandelbaum \cite{9Mand} introduced a new version of autocorrelation on binary (periodical) sequences, called arithmetic autocorrelation, in his research on coding theory (arithmetic codes). Later, such notion is generalized to non-binary case and investigated by Goresky and Klapper \cite{3G1,4G2,5G3,6G4} related to their research on sequences generated by linear shift registers with carry.

Now we state the definition of the arithmetic autocorrelation of (periodical)  binary sequences. For any integer $a\geq1$, we have the 2-adic expression of $a$ as
$$a=a_0+a_12+a_22^2+\cdots+a_r2^r\quad (a_i\in\{0,1\},\ r\geq0).$$
The 2-adic weight of $a$ is defined by
$$w_2(a)=\sum_{\lambda=0}^r a_{\lambda}\ (= \text{the number of} \ a_{\lambda}=1).$$
It is obvious that $1\leq w_2(a)\leq r$ and $w_2(2^la)=w_2(a)$ for all $l\geq1$.

Let $a=(a_{\lambda})_{\lambda=0}^{n-1}$ be a binary sequence with period $n\geq2$, $a_{\lambda}\in\{0,1\}$, $a_{n+\lambda}=a_{\lambda}$. For any $\tau$, $0\leq\tau\leq n-1$, the shift sequence $a^{(\tau)}=(a^{(\tau)}_\lambda)_{\lambda=0}^{n-1}$ of $a$ is defined by $a^{(\tau)}=a_{\lambda+\tau}$. Then $a^{(0)}=a$ and $a^{(\tau)}\ (1\leq\tau\leq n-1)$ are distinct sequences.
Let $$\sigma(a)=\sum_{\lambda=0}^{n-1}a_{\lambda}2^{\lambda}\in\mathbb{Z}.$$
By assumption $n\geq2$, we know that $(a_0,\ \ldots,a_{n-1})\neq(0,\ldots,0),\ (1,\ldots,1)$ and $1\leq\sigma(a)\leq2^n-2$. For $1\leq\tau\leq n-1$, let
$$\sigma(a,\tau)=\sigma(a)-\sigma(a^{(\tau)})=\sum_{\lambda=0}^{n-1}(a_{\lambda}-a_{\lambda+\tau})2^{\lambda}\in\mathbb{Z}.$$
From $1\leq\sigma(a),\sigma(a^{\tau})\leq2^n-2$  and $\sigma(a)\neq\sigma(a^{\tau})$ we get $1\leq\arrowvert \sigma(a)-\sigma(a^{\tau})\arrowvert \leq2^n-3$. Let
$$\sigma(a)-\sigma(a^{\tau})=\varepsilon\cdot\sum_{\lambda=0}^{n-1}c_{\lambda}2^{\lambda}\quad(c_{\lambda}\in\{0,1\},\ \varepsilon\in{\{\pm1\}}).$$
Then $1\leq \sum\limits_{\lambda=0}^{n-1}c_{\lambda}2^{\lambda}=\arrowvert \sigma(a)-\sigma(a^{\tau})\arrowvert  \leq2^n-3$ and $1\leq  w_2(\arrowvert\sigma(a)-\sigma(a^{\tau})\arrowvert)\leq n-1.$

\begin{definition}\label{def1}
	Let $a=(a_{\lambda})_{\lambda=0}^{n-1}$ be a binary sequence with period $n\geq2$. For $1\leq\tau\leq n-1$, let
	$$\sigma(a,\tau)=\sigma(a)-\sigma(a^{(\tau)})=\varepsilon\cdot\sum_{\lambda=0}^{n-1}c_{\lambda}2^{\lambda}\quad(c_{\lambda}\in\{0,1\},\ \varepsilon\in{\{\pm1\}}).$$
	The arithmetic autocorrelation $A_a(\tau)$ is defined by
	
	(1). For $\varepsilon=1\ (\sigma(a)>\sigma(a^{(\tau)})),$
	\begin{align*}
		A_a(\tau)&=\text{``number of $c_{\lambda}=1$"-``number of $c_{\lambda}=0$"}\\
		&=2w_2(\sum_{\lambda=0}^{n-1}c_{\lambda}2^{\lambda})-n=2w_2(\arrowvert \sigma(a,\tau)\arrowvert)-n.
	\end{align*}
	
	(2). For $\varepsilon=-1\ (\sigma(a)<\sigma(a^{(\tau)})),$
	\begin{align*}
		A_a(\tau)&=\text{``number of $c_{\lambda}=0$"-``number of $c_{\lambda}=1$"}\\
		&=n-2w_2(\sum_{\lambda=0}^{n-1}c_{\lambda}2^{\lambda})=n-2w_2(\arrowvert \sigma(a,\tau)\arrowvert).
	\end{align*}
\end{definition}
From $1\leq w_2(\sum_{\lambda=0}^{n-1}c_{\lambda}2^{\lambda})\leq n-1$ we know that $\arrowvert A_a(\tau)\arrowvert\leq n-2$ for all $\tau$, $1\leq\tau\leq n-1$.

The arithmetic autocorrelation distribution of $a$ is the following multiset
$$D(a)=\{A_a(\tau)\arrowvert 1\leq\tau\leq n-1\}.$$

Comparing with the classical autocorrelation, there exist few results currently on arithmetic auto (and cross) correlation on binary sequences (\cite{1C1,2C2,3G1,4G2,5G3,6G4,7H1,8H2,9Mand,10W}). It is desirable that the absolute values $\arrowvert A_a(\tau)\arrowvert\ (1\leq\tau\leq n-1)$ are as small as possible. Hofer, Merai and Winterhof \cite{8H2} proved that for a truly random binary sequence with period $n$, $\arrowvert A_a(\tau)\arrowvert=O(n^{\frac{3}{4}}(\log_2n)^{\frac{1}{2}}$. There exist a family of binary sequences, called {\boldmath$l$}-sequences, having the ideal arithmetic autocorrelation which means that $A_a(\tau)=0\ (1\leq\tau\leq n-1)$ \cite{4G2}. Hofer and Winterhof \cite{7H1} present nontrivial bounds of arithmetic autocorrelations on Legendre sequences. Recently Chen et al. proved that $\arrowvert A_s(\tau)\arrowvert\leq 2^{m-1}-1$ for binary {\boldmath$m$}-sequences with period $n=2^m-1$\ (\cite{1C1}, Theorem 1). Based on examples in cases $m$=3 and 4, they raise the following conjecture.

\begin{conjecture}(\cite{1C1}. Conjecture 1) Let $s$ be a binary {\boldmath$m$}-sequence with period $n=2^m-1,\ m\geq2$. Then
	
	1.\ the arithmetic autocorrelation of $s$ satisfies
	$$A_s(\tau)\in\{\pm(2^k-1):\ 1\leq k\leq m-1\}\quad \text{for}\ 1\leq\tau\leq n-1.$$
	
	2.\ for each $k,\ 1\leq k\leq m-1$, the total number of $\tau\ (1\leq\tau\leq n-1)$ with $\arrowvert A_s(\tau)\arrowvert=2^k-1$ is $2^{n-k}.$
\end{conjecture}
In this paper we prove this conjecture. Firstly we show several preliminary results on $A_a(\tau)$ and present a formula of $A_a(\tau)$ for general (periodic) binary sequences in Section 2 (Theorem 1). Then we totally determine the arithmetic autocorrelation distribution of all binary {\boldmath$m$}-sequences (Theorem 2) in Section 3. As a direct consequence, the above conjecture is verified.

\section{A general formula of arithmetic autocorrelation}\label{sec2}
In this section we present a formula on arithmetic autocorrelation $A_a(\tau)\ (1\leq\tau\leq n-1)$ for arbitrary binary sequence $a=(a_{\lambda})_{\lambda=0}^{n-1}$ with period $n$. Firstly we need some observations.

Let $a=(a_{\lambda})_{\lambda=0}^{n-1}$ and $b=(b_{\lambda})_{\lambda=0}^{n-1}$ be two distinct binary sequences with period $n\geq2$. We have $2\times n\ \{0,1\}$-matrix
$$M=\left(
\begin{array}{cc}
	a_0\ a_1\ \ldots\ a_{n-1}\\ b_0\ b_1\ \ldots\ b_{n-1}
\end{array}
\right)
\ \ \quad (a_{\lambda},\ b_{\lambda}\in\{0,1\})
$$
and (cyclic) shift matrices of $M$
$$M^{(t)}=\left(
\begin{array}{cc}
	a_t\ a_{t+1}\ \ldots\ a_{t+n-1}\\ b_t\ b_{t+1}\ \ldots\ b_{t+n-1}
\end{array}
\right)
\ \ \quad (0\leq t\leq n-1)
$$
where $a_{n+l}=a_l$, $b_{n+l}=b_l$ and $M^{(0)}=M$. Let
$$\sigma(a^{(t)})=\sum_{\lambda=0}^{n-1}a_{\lambda+t}2^{\lambda},\ \sigma(b^{(t)})=\sum_{\lambda=0}^{n-1}b_{\lambda+t}2^{\lambda}\in\mathbb{Z}.$$
From $a\neq b$ we know that $\sigma(a^{(t)})\neq\sigma(b^{(t)}).$ Let
\begin{align}\label{1}
	A(M^{(t)})=\left\{ \begin{array}{ll}
		n-2w_2(\arrowvert\sigma(a^{(t)})-\sigma(b^{(t)})\arrowvert),                             & \mbox{ if}\ \sigma(a^{(t)})>\sigma(b^{(t)}) \\
		2w_2(\arrowvert\sigma(a^{(t)})-\sigma(b^{(t)})\arrowvert)-n,                             & \mbox{ if}\ \sigma(a^{(t)})<\sigma(b^{(t)}).
	\end{array}
	\right.
\end{align}
\begin{lemma}\label{le2.1}
	With above assumptions and notations, we have
	$$A(M)=A(M^{(t)})\quad \text{for all}\ 1\leq t\leq n-1.$$
\end{lemma}
\begin{proof}
	We just need to show that $A(M)=A(M^{(1)})$ since the general case can be done from $M^{(t+1)}=(M^{(t)})^{(1)}$. We have
	$$M=\left(
	\begin{array}{cc}
		a_0\ a_1\ \ldots\ a_{n-1}\\ b_0\ b_1\ \ldots\ b_{n-1}
	\end{array}
	\right),
	\ \ \quad M^{(1)}=\left(
	\begin{array}{cc}
		a_1\ \ldots\ a_{n-1}\ a_0 \\ b_1\ \ldots\ b_{n-1}\ b_0
	\end{array}
	\right)
	$$
	and
	\begin{align}\label{2}
		\left. \begin{array}{ll}
			~~~~~~\sigma(a)-\sigma(b)&=\sum\limits_{\lambda=0}^{n-1}(a_{\lambda}-b_{\lambda})2^{\lambda}=(a_{0}-b_{0})+\sum\limits_{\lambda=1}^{n-1}(a_{\lambda}-b_{\lambda})2^{\lambda}\\&=(a_{0}-b_{0})+2\cdot\sum\limits_{\lambda=0}^{n-2}(a_{\lambda+1}-b_{\lambda+1})2^{\lambda}\\
			\sigma(a^{(1)})-\sigma(b^{(1)})&=\sum\limits_{\lambda=0}^{n-1}(a_{\lambda+1}-b_{\lambda+1})2^{\lambda}=\sum\limits_{\lambda=0}^{n-2}(a_{\lambda+1}-b_{\lambda+1})2^{\lambda}+(a_{0}-b_{0})2^{n-1}
		\end{array}
		\right\}
	\end{align}
	$(\Rmnum{1})$ Case $a_0=b_0.$ By (\ref{2}) we have
	$$\sigma(a)-\sigma(b)=2\cdot\sum\limits_{\lambda=0}^{n-2}(a_{\lambda+1}-b_{\lambda+1})2^{\lambda}=2(\sigma(a^{(1)})-\sigma(b^{(1)})).$$
	Therefore $\sigma(a)-\sigma(b)$ and $\sigma(a^{(1)})-\sigma(b^{(1)})$ have the same sign and $w_2(\arrowvert\sigma(a)-\sigma(b)\arrowvert)=w_2(2\arrowvert\sigma(a)^{(1)}-\sigma(b)^{(1)}\arrowvert)=w_2(\arrowvert\sigma(a)^{(1)}-\sigma(b)^{(1)}\arrowvert).$ From (\ref{1}) we get $A(M)=A(M^{(1)})$.
	
	$(\Rmnum{2})$ Case $(a_0,b_0)=(1,0).$ From (\ref{2}) we have
	\begin{align*}
		\begin{array}{ll}
			\sigma(a^{(1)})-\sigma(b^{(1)})=\sum\limits_{\lambda=0}^{n-2}a_{\lambda+1}2^{\lambda}-\sum\limits_{\lambda=0}^{n-2}b_{\lambda+1}2^{\lambda}+2^{n-1}>0,\\
			\sigma(a)-\sigma(b)=1+\sum\limits_{\lambda=1}^{n-1}(a_{\lambda}-b_{\lambda})2^{\lambda}.
		\end{array}
	\end{align*}
	$(\Rmnum{2}.1)$ If $\sum\limits_{\lambda=1}^{n-1}(a_{\lambda}-b_{\lambda})2^{\lambda}\geq0$, then $\sigma(a)-\sigma(b)>0$. From (\ref{2}) we have
	\begin{align*}
		&w_2(\sigma(a)-\sigma(b))=w_2(1+\sum\limits_{\lambda=1}^{n-1}(a_{\lambda}-b_{\lambda})2^{\lambda})=1+w_2(\sum\limits_{\lambda=1}^{n-1}(a_{\lambda}-b_{\lambda})2^{\lambda}),\\
		&w_2(\sigma(a^{(1)})-\sigma(b^{(1)}))=w_2(\sum\limits_{\lambda=0}^{n-2}(a_{\lambda+1}-b_{\lambda+1})2^{\lambda}+2^{n-1})\\
		&~~~~~~~~~~~~~~~~~~~~~~~~~=w_2(\sum\limits_{\lambda=0}^{n-2}(a_{\lambda+1}-b_{\lambda+1})2^{\lambda})+1=w_2(\sigma(a)-\sigma(b))
	\end{align*}
	From (\ref{1}) we get $A(M)=A(M^{(1)})$.
	
	$(\Rmnum{2}.2)$ If $\sum\limits_{\lambda=1}^{n-1}(a_{\lambda}-b_{\lambda})2^{\lambda}<0$, then $\sum\limits_{\lambda=1}^{n-1}(a_{\lambda}-b_{\lambda})2^{\lambda}=2\sum\limits_{\lambda=1}^{n-1}(a_{\lambda}-b_{\lambda})2^{\lambda}\leq-2$ and $\sigma(a)-\sigma(b)\leq1-2<0$.  From (\ref{2}) we have
	\begin{align*}
		&w_2(\arrowvert\sigma(a)-\sigma(b)\arrowvert)=w_2(-1+\sum\limits_{\lambda=1}^{n-1}(b_{\lambda}-a_{\lambda})2^{\lambda})\\
		&~~~~~~~~~~~~~~~~~~~~~~=n-w_2(2^n-1-(-1+\sum\limits_{\lambda=1}^{n-1}(b_{\lambda}-a_{\lambda})2^{\lambda}))\\
		&~~~~~~~~~~~~~~~~~~~~~~=n-w_2(2^n-(\sum\limits_{\lambda=1}^{n-1}(b_{\lambda}-a_{\lambda})2^{\lambda}))\\
		&w_2(\arrowvert\sigma(a^{(1)})-\sigma(b^{(1)})\arrowvert)=w_2(2^{n-1}-\sum\limits_{\lambda=0}^{n-2}(b_{\lambda+1}-a_{\lambda+1})2^{\lambda})\\
		&~~~~~~~~~~~~~~~~~~~~~~~~~~=w_2(2^n-\sum\limits_{\lambda=1}^{n-1}(b_{\lambda}-a_{\lambda})2^{\lambda})=n-w_2(\arrowvert\sigma(a)-\sigma(b)\arrowvert).
	\end{align*}
	From (\ref{1}) we get
	\begin{align*}
		A(M^{(1)})&=n-2w_2(\arrowvert\sigma(a^{(1)})-\sigma(b^{(1)})\arrowvert)=n-2(n-w_2(\arrowvert\sigma(a)-\sigma(b)\arrowvert))\\
		&=2w_2(\arrowvert\sigma(a)-\sigma(b)\arrowvert)-n=A(M).
	\end{align*}
	
	$(\Rmnum{3})$ Case $(a_0,b_0)=(0,1).$ Let
	$$\widetilde{M}=\left(
	\begin{array}{cc}
		b_0\ \ldots\ b_{n-1}\\a_0\ \ldots\ a_{n-1}
	\end{array}
	\right),
	\ \ \quad \widetilde{M}^{(1)}=\left(
	\begin{array}{cc}
		b_1\ \ldots\ b_{n-1}\ b_0\\a_1\ \ldots\ a_{n-1}\ a_0
	\end{array}
	\right).
	$$
	
	From the formula (\ref{1}) we know that $A(M)=-A(\widetilde{M})$,\ $A(M^{(1)})=-A(\widetilde{M}^{(1)})$. By Case $(\Rmnum{2})$ we get that $A(\widetilde{M})=A(\widetilde{M}^{(1)})$. Therefore $A(M)=A(M^{(1)})$.
	
	This completes the proof of Lemma \ref{le2.1}.
\end{proof}
As a direct consequence of Lemma \ref{le2.1}, we obtain the following result which means that for fixed $\tau$, the value $A_a(\tau)$ is a shift-invariant of binary sequence $a$.
\begin{lemma}\label{le2.2}
	Let $a=(a_{\lambda})_{\lambda=0}^{n-1}$ be a binary sequence with period $n\geq2$. Then for any $t$, $1\leq t\leq n-1$, $A_a(\tau)=A_{a^{(t)}}(\tau)$.
\end{lemma}
\begin{proof}
	Take $b=a^{(\tau)}$ in  Lemma \ref{le2.1}. It is easy to see that $A(M)=A_a(\tau)$ and $A(M^{(t)})=A_{a^{(t)}}(\tau)$. Then from $A(M)=A(M^{(t)})$\ (Lemma \ref{le2.1}) we get $A_a(\tau)=A_{a^{(t)}}(\tau)$.
\end{proof}
We continue to compute the value of $A(M)$ where
$$M=\left(
\begin{array}{cc}
	a_0 \ \ldots\ a_{n-1}\\ b_0 \ \ldots\ b_{n-1}
\end{array}
\right)$$
and $a=(a_{\lambda})_{\lambda=0}^{n-1}$, $b=(b_{\lambda})_{\lambda=0}^{n-1}$ are distinct binary sequences with period $n\geq2$. Now we put a weak assumption.
$$(*)\ \text{There\ exists}\ t,\ 0\leq t\leq n-1\ \text{such\ that}\ (a_t,b_t)=(1,0).$$
Namely, the matrix $M$ have at least one column ${a_t \choose b_t}={1 \choose 0}$. Then the shift matrix $W=M^{(t+1)}$ is
$$W=\left(
\begin{array}{cc}
	a_{t+1}\ \ldots\ a_{t-1}\ a_t\\ b_{t+1}\ \ldots\ b_{t-1}\ b_t
\end{array}
\right)=\left(
\begin{array}{cc}
	a_{t+1}\ \ldots\ a_{t-1}\ 1\\ b_{t+1}\ \ldots\ b_{t-1}\ 0
\end{array}
\right).$$
In this case, $\sigma=\sum\limits_{\lambda=0}^{n-1}a_{\lambda+t+1}2^{\lambda}=\sum\limits_{\lambda=0}^{n-2}a_{\lambda+t+1}2^{\lambda}+2^{n-1}$ is bigger than $\sigma'=\sum\limits_{\lambda=0}^{n-1}b_{\lambda+t+1}2^{\lambda}=\sum\limits_{\lambda=0}^{n-2}b_{\lambda+t+1}2^{\lambda}$, we get
$A(W)=n-2w_2(\sigma-\sigma')$ and $A(M)=A(M^{(t+1)})=A(W)$ by Lemma \ref{le2.1}.

In order to give a closed formula for $A(W)$, we introduce some notations.
\begin{definition}\label{def2}
	For $\alpha,\ \beta,\ \gamma_1,\ \ldots,\ \gamma_l\in\{0,1\}$ and $l\geq0$, we call the matrix
	$$\left(
	\begin{array}{cc}
		\alpha\ \ \  \gamma_1\ \ldots\ \gamma_l\ \ \ \ \beta\\ 1-\alpha\ \gamma_1\ \ldots\ \gamma_l\ 1-\beta
	\end{array}
	\right)$$
	as one with type $[\alpha,\ \beta;l]$.
\end{definition}
\begin{theorem}\label{th1}
	Let $a=(a_{\lambda})_{\lambda=0}^{n-1}$ and $b=(b_{\lambda})_{\lambda=0}^{n-1}$ be distinct binary sequences with period $n\geq2$, and
	$$M=\left(
	\begin{array}{cc}
		a_0\ \ldots\ a_{n-2}\ a_{n-1}\\ b_0\ \ldots\ b_{n-2}\ b_{n-1}
	\end{array}
	\right).$$
	Suppose that there exists $\lambda$, $0\leq\lambda\leq n-1$ such that $(a_{\lambda},b_{\lambda})=(1,0)$. Then
	$$A(M)=n-2g(M)$$
	where $$g(M)=\sum\limits_{l=0}^{n-1}l(N(0,\ 0;l)+N(0,\ 1;l))+\sum\limits_{l=0}^{n-1}(N(1,\ 0;l)+N(1,\ 1;l))$$ and for $\alpha,\ \beta\in\{0,1\}$, $N(\alpha,\ \beta;l)$ is the number of matrices $\left(
	\begin{array}{cc}
		a_{\lambda}\ a_{\lambda+1}\ \ldots\ a_{\lambda+l+1}\\ b_{\lambda}\ b_{\lambda+1}\ \ldots\ b_{\lambda+l+1}
	\end{array}
	\right)$\ $(0\leq\lambda\leq n-1)$ having type $[\alpha,\ \beta;l]$\ $($namely,\ $\left(
	\begin{array}{cc}
		a_{\lambda}\ a_{\lambda+1}\ \ldots\ a_{\lambda+l+1}\\ b_{\lambda}\ b_{\lambda+1}\ \ldots\ b_{\lambda+l+1}
	\end{array}
	\right)=\left(
	\begin{array}{cc}
		\alpha\ \ \  \gamma_1\ \ldots\ \gamma_l\ \ \ \ \beta\\ 1-\alpha\ \gamma_1\ \ldots\ \gamma_l\ 1-\beta
	\end{array}
	\right)$ for some $\gamma_1,\ \ldots,\ \gamma_l\in\{0,1\})$.
\end{theorem}
\begin{proof}
	By assumption $(a_{\lambda},b_{\lambda})=(1,0)$ for some $\lambda$, we have a (cyclic) shift matrix $W$ of $M$
	$$W=\left(
	\begin{array}{ll}
		\alpha_0 \ldots\  \alpha_{n-1}\\
		\varepsilon_0 \ \ldots\  \varepsilon_{n-1}
	\end{array}
	\right)$$
	such that $\left(
	\begin{array}{cc}
		\alpha_{n-1}\\ \varepsilon_{n-1}
	\end{array}
	\right)=\left(
	\begin{array}{cc}
		a_{\lambda}\\ b_{\lambda}
	\end{array}
	\right)=\left(
	\begin{array}{cc}
		1\\ 0
	\end{array}
	\right)$. Let $\sigma(\alpha)=\sum\limits_{\lambda=0}^{n-1}\alpha_{\lambda}2^{\lambda}$, $\sigma(\varepsilon)=\sum\limits_{\lambda=0}^{n-1}\varepsilon_{\lambda}2^{\lambda}$. From $\alpha_{n-1}=1,\ \varepsilon_{n-1}=0$ we know that $\sigma(\alpha)\geq 2^{n-1}>\sigma(\varepsilon).$ Since $W$ is shift  matrix of $M$, we get
	$$A(M)=A(W)=n-2w_2(\sigma(\alpha)-\sigma(\varepsilon)).$$
	We compute $\sigma(\alpha)-\sigma(\varepsilon)$ by 2-adic subtraction on the digits of $\sigma(\alpha)$ and $\sigma(\varepsilon)$ as shown in Figure 1.
	\begin{center}
		\begin{minipage}{200pt}
		\ \ \ \	\textbf{	Figure 1} \ \ 2-adic subtraction $\sigma(\alpha)-\sigma(\varepsilon)$\\
			~\\
			\begin{tabular}{@{}c|lllll@{}}\label{fi1}
				$\sigma(\alpha)$ & $\alpha_0$ & $\alpha_1$  & $\ldots$ &  $\alpha_{n-1}$& $(\alpha_{n-1}=1)$\\
				$\sigma(\varepsilon)$    & $\varepsilon_0$ & $\varepsilon_1$  & $\ldots$ &  $\varepsilon_{n-1}$& $(\varepsilon_{n-1}=0)$\\	\hline
				$\sigma(\alpha)-\sigma(\varepsilon)$    & $\delta_0$ & $\delta_1$  & $\ldots$ &  $\delta_{n-1}$& $(\delta_{\lambda}\in\{0,1\})$
			\end{tabular}
		\end{minipage}
	\end{center}
	Then $\sigma(\alpha)-\sigma(\varepsilon)=\sum\limits_{\lambda=0}^{n-1}\delta_{\lambda}2^{\lambda}$ and $w_2(\sigma(\alpha)-\sigma(\varepsilon))$ is the number of $\delta_{\lambda}=1\
	(0\leq\lambda\leq n-1).$
	
	Now we write the matrix $W$ in the following block form
	$$W=\left(W_1\ W_2\ \ldots\ W_m\right)\quad m\geq1.$$
	Each bloack $W_i$ has form
	$$W_i=\left(
	\begin{array}{ll}
		\gamma_1^{(i)}\ \ldots\ \gamma_{l
			_i}^{(i)}\ \ \ \ \beta_i\\ \gamma_1^{(i)}\ \ldots\ \gamma_{l
			_i}^{(i)}\ 1-\beta_i
	\end{array}
	\right)\quad(l_i\geq0)$$
	where $\gamma_1^{(i)},\ldots,\gamma_{l
		_i}^{(i)}\in\{0,1\}$ and the last block is $W_m=\left(
	\begin{array}{cc}
		\gamma_1^{(m)}\ \ldots\ \gamma_{l
			_m}^{(m)}\ 1\\ \gamma_1^{(m)}\ \ldots\ \gamma_{l
			_m}^{(m)}\ 0
	\end{array}
	\right)\ (\beta_m=\alpha_{n-1}=1).$
	Put the last column of $W_{i-1}$ to the left-hand side of $W_i$, we get
	$$\widetilde{W_i}=\left(
	\begin{array}{cc}
		\beta_{i-1}\ \ \  	\gamma_1^{(i)}\ \ldots\ \gamma_{l
			_i}^{(i)}\ \ \ \ \beta_i\\ 1-\beta_{i-1}\ \gamma_1^{(i)}\ \ldots\ \gamma_{l
			_i}^{(i)}\ 1-\beta_i
	\end{array}
	\right)=\left(
	\begin{array}{cc}
		\beta_{i-1}	\\ 1-\beta_{i-1}
	\end{array}W_i
	\right)$$
	which has type $[\beta_{i-1},\ \beta_i;l_i]$. For the first block $W_1$, we assume $\widetilde{W_1}=\left(
	\begin{array}{cc}
		\beta_{m}	\\ 1-\beta_{m}
	\end{array}W
	\right)=\left(
	\begin{array}{ll}
		1\ \gamma_1^{(1)}\ \ldots\ \gamma_{l
			_1}^{(1)}\ \ \ \ \beta_1\\ 0\ \gamma_1^{(1)}\ \ldots\ \gamma_{l
			_1}^{(1)}\ 1-\beta_1
	\end{array}
	\right)$. Then we can see the 2-adic subtraction given in Figure 1 locally. For each $\widetilde{W_i}$ with type $[\beta_{i-1},\ \beta_i;l_i]$, the digits $\delta_{\lambda}$ (in Figure 1) under the columns of $W_i$ can be determined as shown in Figure 2.
	\begin{center}
		\textbf{	Figure 2} \ \ 2-adic subtraction at $W_i$ \\
		~\\
		\begin{minipage}{.45\linewidth}
			\begin{center}
				$(\beta_{i-1},\beta_i)=(1,0)$\\~\\
				\begin{tabular}{@{}cccccc@{}}
					(1) & $\gamma_1^{(i)}$  & $\ldots$ &  $\gamma_{l
						_i}^{(i)}$& 0 & $\ldots$ \\
					(0)   & $\gamma_1^{(i)}$  & $\ldots$ &  $\gamma_{l
						_i}^{(i)}$& 1 & $\ldots$ \\	\hline
					$\ldots$ & 0 & $\ldots$  & 0 &  1 & $\ldots$
			\end{tabular}	\end{center}
		\end{minipage}
		\begin{minipage}{.45\linewidth}
			\begin{center}
				$(\beta_{i-1},\beta_i)=(0,0)$\\~\\
				\begin{tabular}{@{}cccccc@{}}
					(0) & $\gamma_1^{(i)}$  & $\ldots$ &  $\gamma_{l
						_i}^{(i)}$& 0 & $\ldots$ \\
					(1)   & $\gamma_1^{(i)}$  & $\ldots$ &  $\gamma_{l
						_i}^{(i)}$& 1 & $\ldots$ \\	\hline
					$\ldots$ & 1 & $\ldots$  & 1 &  0 & $\ldots$
			\end{tabular}	\end{center}
		\end{minipage}\\~\\~\\
		\begin{minipage}{.45\linewidth}
			\begin{center}
				$(\beta_{i-1},\beta_i)=(1,1)$\\~\\
				\begin{tabular}{@{}cccccc@{}}
					(1) & $\gamma_1^{(i)}$  & $\ldots$ &  $\gamma_{l
						_i}^{(i)}$& 1 & $\ldots$ \\
					(0)   & $\gamma_1^{(i)}$  & $\ldots$ &  $\gamma_{l
						_i}^{(i)}$& 0 & $\ldots$ \\	\hline
					$\ldots$ & 0 & $\ldots$  & 0 &  1 & $\ldots$
			\end{tabular}	\end{center}
		\end{minipage}
		\begin{minipage}{.45\linewidth}
			\begin{center}
				$(\beta_{i-1},\beta_i)=(0,1)$\\~\\
				\begin{tabular}{@{}cccccc@{}}
					(0) & $\gamma_1^{(i)}$  & $\ldots$ &  $\gamma_{l
						_i}^{(i)}$& 1& $\ldots$ \\
					(1)   & $\gamma_1^{(i)}$  & $\ldots$ &  $\gamma_{l
						_i}^{(i)}$& 0 & $\ldots$ \\	\hline
					$\ldots$ & 1 & $\ldots$  & 1 &  0 & $\ldots$
			\end{tabular}	\end{center}
		\end{minipage}
	\end{center}
	From Figure 2 we know that the number $\delta_{\lambda}=1$ given by $W_i$ is 1 if the type of $\widetilde{W_i}$ is $[\beta_{i-1},\ \beta_i;l_i]=[1,\ 0;l_i]$ or $[1,\ 1;l_i]$; or $l_i$ if the type of $\widetilde{W_i}$ is $[\beta_{i-1},\ \beta_i;l_i]=[0,\ 0;l_i]$ or $[0,\ 1;l_i]$. Put all local data together, we get
	\begin{align*}
		&w_2(\sigma(\alpha)-\sigma(\varepsilon))=\sum_{\substack{i=1\\ \text{``type of $\widetilde{W_i}$"=$[1,\ \beta_i;l_i]$}\\\beta_i\in\{0,1\}}}^{m}1+\sum_{\substack{i=1\\ \text{``type of $\widetilde{W_i}$"=$[0,\ \beta_i;l_i]$}\\\beta_i\in\{0,1\}}}^{m}l_i\\
		=&\sum_{l\geq0}l(N_W(0,\ 0;l)+N_W(0,\ 1;l))+\sum_{l\geq0}(N_W(1,\ 0;l)+N_W(1,\ 1;l)),
	\end{align*}
	where for $l\geq0$, $\alpha,\ \beta\in\{0,1\}$, $N_W(\alpha,\ \beta;l)$ is the number of $\left(
	\begin{array}{cc}
		\alpha_{\lambda}\ \alpha_{\lambda+1}\ \ldots\ \alpha_{\lambda+l+1}\\ \varepsilon_{\lambda}\ \varepsilon_{\lambda+1}\ \ldots\ \varepsilon_{\lambda+l+1}
	\end{array}
	\right)\ (0\leq\lambda\leq n-1)$ having type $[\alpha,\ \beta;l]$. Since $W$ is a shift of $M$, the number $N_W(\alpha,\ \beta;l)$ is the same as the number $N(\alpha,\ \beta;l)=N_M(\alpha,\ \beta;l)$. Therefore
	$$w_2(\sigma(\alpha)-\sigma(\varepsilon))=\sum_{l\geq0}l(N(0,\ 0;l)+N(0,\ 1;l))+\sum_{l\geq0}(N(1,\ 0;l)+N(1,\ 1;l)).$$
	At last, if $l\geq n$ and there exists a matrix $\left(
	\begin{array}{cc}
		a_{\lambda}\ a_{\lambda+1}\ \ldots\ a_{\lambda+l}\ a_{\lambda+l+1}\\ b_{\lambda}\ b_{\lambda+1}\ \ldots\ b_{\lambda+l}\ b_{\lambda+l+1}
	\end{array}
	\right)$ of type $[\alpha,\ \beta;l]$ for some $\lambda$, then $a_{\lambda+i}=b_{\lambda+i}\ (i=1,\ldots,n)$ which means $a=b$ since the period of both sequences $a$ and $b$ is $n$. This contradicts to the assumption $a\neq b$. Therefore $N(\alpha,\ \beta;l)=0$ for any $\alpha,\ \beta\in\{0,1\}$ and  $l\geq n$. We get the final formula
	\begin{align*}
		&\sum_{l\geq0}l(N(0,\ 0;l)+N(0,\ 1;l))+\sum_{l\geq0}(N(1,\ 0;l)+N(1,\ 1;l))\\
		=&\sum_{l=0}^{n-1}l(N(0,\ 0;l)+N(0,\ 1;l))+\sum_{l=0}^{n-1}(N(1,\ 0;l)+N(1,\ 1;l))=g(M)
	\end{align*}
	and $A(M)=n-2g(M)$. This completes the proof of Theorem 1.
\end{proof}

By Theorem \ref{th1}, the value of $A(M)$ is expressed by the numbers $N(1,\ 0;l)+N(1,\ 1;l)$ and $N(0,\ 0;l)+N(0,\ 1;l)$\ ($0\leq l\leq n-1,\ N(\alpha,\ \beta;l)=N_M(\alpha,\ \beta;l)$) for $M=\left(
\begin{array}{cc}
	a_{0}\ \ldots\ a_{n-1}\\ b_{0}\ \ldots\ b_{n-1}
\end{array}
\right)$. To determine these numbers is not easy for arbitrary binary sequences $a=(a_{\lambda})_{\lambda=0}^{n-1}$ and $b=(b_{\lambda})_{\lambda=0}^{n-1}$ with period $n\geq2$ in general case. Fortunately, if $a$ is a binary {\boldmath$m$}-sequence and $b$ is a shift sequence of $a$, these numbers can be computed and then the arithmetic autocorrelation distribution of all binary {\boldmath$m$}-sequences can be determined. We do this in next section.

\section{Arithmetic autocorrelation distribution of binary $m$-sequences}\label{sec3}

Let $q=2^m\ (m\geq2)$, $\pi$ be a primitive element of the finite field $\mathbb{F}_q$, which means that $\mathbb{F}^{\ast}_q=\mathbb{F}_q\setminus\{0\}=\langle\pi\rangle$, $\pi$ is the generator of the cyclic multiplicative group $\mathbb{F}^{\ast}_q$.
The binary {\boldmath$m$}-sequence $s=(s_{\lambda})^{n-1}_{\lambda=0}$ with period $n=2^m-1$ is defined by
$s_{\lambda}=T(\pi^{\lambda})\in\mathbb{F}_2\ (0\leq\lambda\leq n)$, where
$$T:\ \mathbb{F}_q\rightarrow\mathbb{F}_2,\ T(\alpha)=\alpha+\alpha^2+\alpha^{2^2}+\cdots+\alpha^{2^{m-1}}$$
is the trace mapping.

In Theorem \ref{th1}, we take $a=s$, $b=s^{(\tau)}=(s^{(\tau)}_{\lambda})_{\lambda=0}^{n-1}$, the shift sequence of $s$ where $s^{(\tau)}_{\lambda}=s_{\lambda+\tau}$. Then
$M=\left(
\begin{array}{ll}
	s_0\ \ \  s_1\ \ \  \ldots\ s_{n-1}\\ s^{(\tau)}_0\ s^{(\tau)}_1\ \ldots\ s^{(\tau)}_{n-1}
\end{array}
\right)=
\left(
\begin{array}{ll}
	s_0\ \ \  s_1\ \  \ldots\ s_{n-1}\\ s_{\tau}\ s_{\tau+1}\ \ldots\ s_{n-1+\tau}
\end{array}
\right)
\ (s_{n+t}=s_t)
$
and  the arithmetic autocorrelation $A_s(\tau)$ is $A(M)$ for $1\leq\tau\leq n-1$. From Theorem \ref{th1} we know that $A_s(\tau)=n-2g(\tau)$ where
\begin{align}\label{3}
	g(\tau)=\sum\limits_{l=0}^{n-1}l(N^{(\tau)}(0,\ 0;l)+N^{(\tau)}(0,\ 1;l))+\sum\limits_{l=0}^{n-1}(N^{(\tau)}(1,\ 0;l)+N^{(\tau)}(1,\ 1;l))
\end{align}
and for $\alpha,\ \beta\in\mathbb{F}_2=\{0,1\}$
\begin{align*}
	&N^{(\tau)}(\alpha,\ \beta;l)\\
	=& \sharp\left\{0\leq t\leq n-1\Bigg\arrowvert
	\begin{array}{ll}
		\left(
		\begin{array}{ll}
			s_t\ \ \  s_{t+1}\ \ldots\ s_{t+l+1}\\ s^{(\tau)}_t\ s^{(\tau)}_{t+1}\ \ldots\ s^{(\tau)}_{t+l+1}
		\end{array}
		\right)=
		\left(
		\begin{array}{cc}
			\alpha\ \ \  \gamma_1\ \ldots\ \gamma_l\ \ \ \ \beta\\ 1-\alpha\ \gamma_1\ \ldots\ \gamma_l\ 1-\beta
		\end{array}
		\right),\
		\gamma_1,\ldots,\gamma_l\in\mathbb{F}_2	
	\end{array} \right\}\\
	=& \sharp\left\{0\leq t\leq n-1\Bigg\arrowvert \begin{array}{ll}
		T(\pi^t)=\alpha,\ T(\pi^{t+\tau})=1+\alpha,\ T(\pi^{t+l+1})=\beta,\ T(\pi^{t+l+1+\tau})=1+\beta,\\
		T(\pi^{t+\lambda})= T(\pi^{t+\lambda+\tau})\ \text{for}\ 1\leq\lambda\leq l
	\end{array}\right\}\\
	=& \sharp\left\{x\in\mathbb{F}^{\ast}_q\Bigg\arrowvert \begin{array}{ll}
		T(x)=\alpha,\ T(\pi^{\tau}x)=1+\alpha,\ T(\pi^{l+1}x)=\beta,\ T(\pi^{l+1+\tau}x)=1+\beta,\\
		T(\pi^{\lambda}x)= T(\pi^{\lambda+\tau}x)\ \text{for}\ 1\leq\lambda\leq l
	\end{array}\right\}.
\end{align*}
In order to determine $A_s(\tau)$ by formula (\ref{3}), we need to compute
\begin{align}\label{4}
	&N^{(\tau)}(0,\ 0;l)+	N^{(\tau)}(0,\ 1;l)\notag\\
	=& \sharp\left\{x\in\mathbb{F}^{\ast}_q\Bigg\arrowvert \begin{array}{ll}
		T(x)=0,\ T(\pi^{\tau}x)=1,\ T(\pi^{l+1}x)+ T(\pi^{l+1+\tau}x)=1,\\
		T(\pi^{\lambda}x)+T(\pi^{\lambda+\tau}x)=0\ \text{for}\ 1\leq\lambda\leq l
	\end{array}\right\}
\end{align}
and
\begin{align}\label{5}
	&N^{(\tau)}(1,\ 0;l)+	N^{(\tau)}(1,\ 1;l)\notag\\
	=& \sharp\left\{x\in\mathbb{F}^{\ast}_q\Bigg\arrowvert \begin{array}{ll}
		T(x)=1,\ T(\pi^{\tau}x)=0,\ T(\pi^{l+1}x)+ T(\pi^{l+1+\tau}x)=1,\\
		T(\pi^{\lambda}x)+T(\pi^{\lambda+\tau}x)=0\ \text{for}\ 1\leq\lambda\leq l
	\end{array}\right\}.
\end{align}

\begin{lemma}\label{le3.1}
	Let $1\leq\tau\leq n-1.$ Then
	
	$(\Rmnum{1}).$\ If $l\geq m$, then $N^{(\tau)}(\alpha,\ \beta;l)=0$ for any $\alpha,\ \beta\in\mathbb{F}_2.$
	
	$(\Rmnum{2}).\begin{array}{ll} \sum\limits_{l=0}^{m-1}(N^{(\tau)}(1,\ 0;l)+	N^{(\tau)}(1,\ 1;l))=\sum\limits_{l=0}^{m-1}(N^{(\tau)}(0,\ 0;l)+	N^{(\tau)}(0,\ 1;l))=2^{m-2}.
	\end{array}$
	
	$(\Rmnum{3}).$\ $g(\tau)=2^{m-2}+\sum\limits_{l=1}^{m-1}l(N^{(\tau)}(0,\ 0;l)+	N^{(\tau)}(0,\ 1;l)).$
\end{lemma}
\begin{proof}
	$(\Rmnum{1})$\ Suppose that $l\geq m$ and $N^{(\tau)}(\alpha,\ \beta;l)\geq1$. Then there exists $t$ such that $s_{t+i}=s_{t+i}^{(\tau)}$ for $1\leq i\leq l$. From $1\leq\tau\leq n-1$ we know that $s\neq s^{(\tau)}$ and $s+s^{\tau}=(s_{\lambda}+s_{\lambda}^{(\tau)})_{\lambda=0}^{n-1}$ is a shift sequence of $s$\ (Lemma \ref{le1}. (1)). But there is no pattern $(s_{t+1}+s_{t+1}^{(\tau)},\ldots,s_{t+l}+s_{t+l}^{(\tau)})=(0,\ldots,0)$ in sequence $s+s^{\tau}$ for $l\geq m$\ (Lemma \ref{le1}. (2)). This contradiction shows that $N(\alpha,\ \beta;l)=0$ for $l\geq m$ and any $\alpha,\ \beta\in\mathbb{F}_2.$
	
	$(\Rmnum{2})$\ Since the {\boldmath$m$}-sequence $s$ has no pattern $(0,\ldots,0)\in\mathbb{F}_2^m$, from $(\Rmnum{1})$ and formula (\ref{5}) we can see that
	\begin{align}\label{6}
		&\sum_{l=0}^{m-1}N^{(\tau)}(1,\ 0;l)+N^{(\tau)}(1,\ 1;l)\notag\\
		=& \sharp\left\{x\in\mathbb{F}^{\ast}_q\Bigg\arrowvert T(x)=1,\  T(\pi^{\tau}x)=0\right\}\notag\\
		=&\frac{1}{4}\sum_{x\in\mathbb{F}_q^{\ast}}\big(1-(-1)^{T(x)}\big)\big(1+(-1)^{T(\pi^{\tau}x)}\big)\notag\\
		=&\frac{1}{4}\sum_{x\in\mathbb{F}_q}\big(1-(-1)^{T(x)}\big)\big(1+(-1)^{T(\pi^{\tau}x)}\big)\quad(\text{since}\ 1-(-1)^{T(0)}=1-1=0)\notag\\
		=&\frac{1}{4}\left(q+\sum_{x\in\mathbb{F}_q}\left[(-1)^{T(\pi^{\tau}x)}-(-1)^{T(x)}-(-1)^{T((1+\pi^{\tau})x)}\right]\right)
	\end{align}
	It is well-known that $\sum\limits_{x\in\mathbb{F}_q}(-1)^{T(\alpha x)}=\sum\limits_{x\in\mathbb{F}_q}(-1)^{T(x)}=0$ if $\alpha\in\mathbb{F}_q^{\ast}$. Then by $1+\pi^{\tau}\neq0$\ (for $1\leq\tau\leq n-1$) we know that the summation at right-hand side of (\ref{6}) is zero. Therefore $\sum\limits_{l=0}^{m-1}(N^{(\tau)}(1,\ 0;l)+N^{(\tau)}(1,\ 1;l))=\frac{q}{4}
	=2^{m-2}$.
	
	We can show $\sum\limits_{l=0}^{m-1}(N^{(\tau)}(0,\ 0;l)+N^{(\tau)}(0,\ 1;l))=2^{m-2}$ in similar way.
	
	$(\Rmnum{3})$\ Directly from $(\Rmnum{1})$,\ $(\Rmnum{2})$ and formula (\ref{3}).
\end{proof}
Now we come to place of computing $N^{(\tau)}(0,\ 0;l)+N^{(\tau)}(0,\ 1;l)$
for each $l$,\ $1\leq l\leq m-1$ and $1\leq\tau\leq n-1$. For $1\leq\tau\leq n-1$,\ $1+\pi^{\tau}\neq0,\ 1$, because the order of $\pi$ is $n=q-1$. Since $\{1,\pi,\ldots,\pi^{m-1}\}$ is an $\mathbb{F}_2$-basis of $\mathbb{F}_q$, $(1+\pi^{\tau})^{-1}(\neq0,\ 1)$ has the following unique expression
$$(1+\pi^{\tau})^{-1}=b_0+b_1\pi+\cdots+b_{s-1}\pi^{e-1}+\pi^e$$
where $b_0,\cdots,b_{s-1}\in\mathbb{F}_2$, $1\leq e\leq m-1$. Another thing we need is that the minimal polynomial of $\pi$ over $\mathbb{F}_2$ is an irreducible (primitive) polynomial $f(x)=1+a_1x+\cdots+a_{m-1}x^{m-1}+x^m\ (a_1,\ldots,a_{m-1}\in\mathbb{F}_2)$ and $f(\pi)=0$. Therefore $\pi^m=1+a_1\pi+\cdots+a_{m-1}\pi^{m-1}$.
\begin{lemma}\label{le4}
	Let $q=2^m\ (m\geq2)$, $\pi$ be a primitive element of $\mathbb{F}_q$,\ $s=(s_{\lambda})^{n-1}_{\lambda=0}$  is the binary {\boldmath$m$}-sequence with period $n=q-1$ defined by $s_{\lambda}=T(\pi^{\lambda})$, where $T:\ \mathbb{F}_q\rightarrow\mathbb{F}_2$ is the trace mapping. Let $1\leq\tau\leq n-1$ and
	$$(1+\pi^{e})^{-1}=b_0+b_1\pi+\cdots+b_{s-1}\pi^{e-1}+\pi^e\quad(b_0,\ldots,b_{s-1}\in\mathbb{F}_2,\ 1\leq e\leq m-1)$$
	Then
	
	$(\Rmnum{1}).$\ For $1\leq l\leq m-2$,
	\begin{align*}
		N^{(\tau)}(0,\ 0;l)+N^{(\tau)}(0,\ 1;l)=\left\{
		\begin{array}{ll}
			2^{m-l-3},\ &\text{if}\ l\leq e-2\\
			2^{m-l-3}(1-(-1)^{b_0}),\ &\text{if}\ l= e-1(\geq1)\\
			2^{m-l-3}(1+(-1)^{b_0}),\ &\text{if}\ l\geq e.
		\end{array}\right.
	\end{align*}
	
	$(\Rmnum{2}).$\ For $l=m-1$,
	$$N^{(\tau)}(0,\ 0;m-1)+N^{(\tau)}(0,\ 1;m-1)=\frac{1}{2}(1+(-1)^{b_0}).$$
\end{lemma}
\begin{proof}
	From formula (\ref{4}) we know that
	\begin{align*}
		N^{(\tau)}(0,\ 0;&l)+N^{(\tau)}(0,\ 1;l)\\
		=\frac{1}{2^{l+3}}\sum_{x\in\mathbb{F}_q^{\ast}}&\big((1+(-1)^{T(x)})(1-(-1)^{T(\pi^{\tau}x)})(1-(-1)^{T(\pi^{l+1}x+\pi^{l+1+\tau}x)})\prod\limits_{\lambda=1}^l(1+(-1)^{T(\pi^{\lambda}x+\pi^{\lambda+\tau}x)})\big)\\
		=\frac{1}{2^{l+3}}\sum_{x\in\mathbb{F}_q}&\big((1+(-1)^{T(x)})(1-(-1)^{T(\pi^{\tau}x)})(1-(-1)^{T(\pi^{l+1}(1+\pi^{\tau})x)})\prod\limits_{\lambda=1}^l(1+(-1)^{T(\pi^{\lambda}(1+\pi^{\tau})x)})\big)\\
		&~~~~~~~~~~~~~~~~~~~~~~~~~~~~~~~~~~~~~~~~~~~~~~~~~~~~~~~\quad(\text{for}\ x=0,\ 1-(-1)^{T(\pi^{\tau}x)}=0)\\
		=\frac{1}{2^{l+3}}\sum_{x\in\mathbb{F}_q}&\sum_{c',c_0,\ldots,c_{l+1}=0}^{1}(-1)^{T[x(c'+c_0\pi^{\tau}+(1+\pi^{\tau})(c_1\pi+c_2\pi^2+\cdots+c_{l+1}\pi^{l+1}))]}(-1)^{c_0+c_{l+1}}\\
		=\frac{1}{2^{l+3}}\sum_{x\in\mathbb{F}_q}&\sum_{c,c_0,\ldots,c_{l+1}=0}^{1}(-1)^{T[x(c+(1+\pi^{\tau})(c_0+c_1\pi+c_2\pi^2+\cdots+c_{l+1}\pi^{l+1}))]}(-1)^{c_0+c_{l+1}}\quad(\text{let}\ c=c'+c_0)
	\end{align*}
	If $\delta=c+(1+\pi^{\tau})(c_0+c_1\pi+c_2\pi^2+\cdots+c_{l+1}\pi^{l+1})\neq 0$, then $\sum_{x\in\mathbb{F}_q}(-1)^{T(\delta x)}=0$. Therefore
	\begin{align}\label{7}
		N^{(\tau)}(0,\ 0;l)+N^{(\tau)}(0,\ 1;l)=\frac{q}{2^{l+3}}\sum_{\substack{c,c_0,\ldots,c_{l+1}=0\\c_0+c_1\pi+\cdots+c_{l+1}\pi^{l+1}=c(1+\pi^{\tau})^{-1}}}^{1}(-1)^{c_0+c_{l+1}}
	\end{align}
	where the summation in the right-hand side is under the condition
	\begin{align}\label{8}
		c_0+c_1\pi+\cdots+c_{l+1}\pi^{l+1}=c(1+\pi^{\tau})^{-1}=c(b_0+b_1\pi+\cdots+b_{e-1}\pi^{e-1}+\pi^e).
	\end{align}
	From now on we write $N^{(\tau)}(0,\ 0;l)+N^{(\tau)}(0,\ 1;l)$ by $N^{(\tau)}(l)$ briefly.
	
	$(\Rmnum{1})$\ For $1\leq l\leq m-2$, $l+1\leq m-1$ and $e\leq m-1$. Since $\{1,\pi,\ldots,\pi^{m-1}\}$ is a basis of $\mathbb{F}_q$ over $\mathbb{F}_2$, the coefficients of $\pi^{\lambda}$ in both sides of formula (\ref{8}) should be equal for each $\lambda$, $0\leq\lambda\leq m-1$.
	
	$(\Rmnum{1}.1)$\ If $l\leq e-2$, then $e\geq l+2$. From the coefficients of $\pi^{e}$ we get $0=c$. Then $c_{\lambda}=cb_{\lambda}=0$ for all  $\lambda$, $0\leq\lambda\leq l+1\ (\leq e-1)$. By formula ($\ref{7}$) we get $N^{(\tau)}(l)=\frac{q}{2^{l+3}}\cdot 1=2^{m-l-3}$.
	
	$(\Rmnum{1}.2)$\ For $l= e-1$, formula ($\ref{8}$) becomes $c_0+c_1\pi+\cdots+c_{l+1}\pi^{l+1}=c(b_0+b_1\pi+\cdots+b_{l}\pi^{l}+\pi^{l+1})$. Therefore $c_{l+1}=c$ and $c_{\lambda}=cb_{\lambda}\ (0\leq\lambda\leq l)$. For $c=0$, we have $c_0=c_1=\cdots=c_{l+1}=0$ and $c_0+c_{l+1}=0$. For $c=1$, we have $c_{l+1}=1$, $c_{\lambda}=b_{\lambda}\ (0\leq\lambda\leq l)$ and $c_0+c_{l+1}=b_0+1$. We obtain $N^{(\tau)}(e-1)=2^{m-l-3}(1+(-1)^{b_0+1})=2^{m-l-3}(1-(-1)^{b_0})$.
	
	$(\Rmnum{1}.3)$\ For $l\geq e$, then $c_{\lambda}=cb_{\lambda}\ (0\leq\lambda\leq e-1)$, $c_e=c$, $c_{\lambda}=0\ (e+1\leq\lambda\leq l+1)$ and $c_0+c_{l+1}=cb_0$. From formula ($\ref{7}$) we get $$\begin{array}{ll}
		N^{(\tau)}(l)=2^{m-l-3}(&~~~~~1~~~~~+~~(-1)^{b_0}~~).\\
		&\small{\text{(for\ c=0)}~~~~~~~~\text{(for\ c=1)}}
	\end{array}$$
	
	$(\Rmnum{2})$\ For $l=m-1$, formula ($\ref{8}$) becomes
	$$c(b_0+b_1\pi+\cdots+b_{e-1}\pi^{e-1}+\pi^e)=c_0+c_1\pi+\cdots+c_{m-1}\pi^{m-1}+c_{m}\pi^{m}.$$
	Let $f(x)=1+a_1x+\cdots+a_{m-1}x^{m-1}+x^m\ (a_i\in\mathbb{F}_2)$
	be the minimal polynomial of $\pi$ over $\mathbb{F}_2$. Then $\pi^m=1+a_1\pi+\cdots+a_{m-1}\pi^{m-1}$ and
	$$\begin{array}{ll}c(b_0+b_1\pi+\cdots+b_{e-1}\pi^{e-1}+\pi^e)
		=c_0+c_1\pi+\cdots+c_{m-1}\pi^{m-1}+c_{m}(1+a_1\pi+\cdots+a_{m-1}\pi^{m-1})\end{array}$$
	which imply that
	$$\begin{array}{ll}
		cb_0=c_0+c_m=c_0+c_{l+1},\ cb_{\lambda}=c_{\lambda}+c_{l+1}a_{\lambda}\ (1\leq\lambda\leq e-1),\ c=c_e+c_{l+1}a_e,\\
		0=c_{\lambda}+c_{l+1}a_{\lambda}\ (e+1\leq\lambda\leq m-1=l).
	\end{array}$$
	For $c=0$ we have $c_0=c_{l+1}$, $c_{\lambda}=c_{l+1}a_{\lambda}\ (1\leq\lambda\leq l)$. For $c=1$ we have $c_0=b_0+c_{l+1}$, $c_{\lambda}=b_{\lambda}+c_{l+1}a_{\lambda}\ (1\leq\lambda\leq e-1)$, $c_e=1+c_{l+1}a_e$, $c_{\lambda}=c_{l+1}a_{\lambda}\ (e+1\leq\lambda\leq l)$. By formula ($\ref{7}$) we get
	\begin{align*}
		&\begin{array}{ll}
			N^{(\tau)}(m-1)=2^{m-l-3}(\sum\limits_{c_{l+1}=0}^{1}(-1)^{c_{l+1}+c_{l+1}}+\sum\limits_{c_{l+1}=0}^{1}(-1)^{b_0+c_{l+1}+c_{l+1}})\\
			~~~~~~~~~~~~~~~~~~~~~~~~~~~~~~~\small{\text{(for\ c=0)}~~~~~~~~~~~~~~~~~~\text{(for\ c=1)}}\end{array}
		=2^{m-l-2}(1+(-1)^{b_0})=\frac{1}{2}(1+(-1)^{b_0}).
	\end{align*}
	This completes the proof of Lemma \ref{le4}.
\end{proof}

Now we show our main result on arithmetic autocorrelation distribution $\{A_s(\tau)\arrowvert1\leq\tau\leq n\}$ of any binary {\boldmath$m$}-sequence with period $n=2^m-1$.
\begin{theorem}\label{th2}
	Let $m\geq2$, $n=2^m-1$, $s=(s_{\lambda}=T(\pi^{\lambda}))^{n-1}_{\lambda=0}$ be the binary {\boldmath$m$}-sequence with period $n$ where $\pi$ is a primitive element of $\mathbb{F}_q\ (q=2^m)$ and $T:\ \mathbb{F}_q\rightarrow\mathbb{F}_2$ is the trace mapping. For $1\leq\tau\leq n-1$, let $A(\tau)=A_s(\tau)$ be the arithmetic autocorrelation of $s$ and
	$$(1+\pi^{\tau})^{-1}=b_0+b_1\pi+\cdots+b_{e-1}\pi^{e-1}+\pi^e\quad(1\leq e\leq m-1,\ b_0,b_1,\ldots,b_{e-1}\in\mathbb{F}_2).$$
	
	Then
	
	$(\Rmnum{1})$. For
	$N^{(\tau)}(l)=N^{(\tau)}(0,\ 0;l)+N^{(\tau)}(0,\ 1;l)$,
	\begin{align*}
		\sum_{l=0}^{m-1}lN^{(\tau)}(l)=\left\{
		\begin{array}{ll}
			2^{m-2}+2^{m-e-1}-1,\ &\text{if}\ b_0=0\\
			2^{m-2}-2^{m-e-1},\ &\text{if}\ b_0=1.\\
		\end{array}\right.
	\end{align*}
	
	$(\Rmnum{2})$.\ (arithmetic autocorrelation distribution)\ For all $\tau$, $1\leq\tau\leq n-1$, $A_s(\tau)\in\{\pm(2^k-1)\arrowvert1\leq k\leq m-1\}$. Moreover, for each $\varepsilon\in\{0,1\}$ and $0\leq k\leq n-1$, the number of $\tau\ (1\leq\tau\leq n-1)$ such that $A_s(\tau)=\varepsilon(2^k-1)$ is $2^{m-k-1}$.
\end{theorem}
\begin{proof}
	For $1\leq\tau\leq n-1$ and $(1+\pi^{\tau})^{-1}=b_0+b_1\pi+\cdots+b_{e-1}\pi^{e-1}+\pi^e\quad(1\leq e\leq m-1,\ b_0,\ldots,b_{e-1}\in\mathbb{F}_2)$, by Lemma \ref{le4} we know that
	\begin{align*}
		\sum_{l=1}^{m-1}lN^{(\tau)}(l)=&\sum_{1\leq l\leq e-2}l\cdot 2^{m-l-3}+(e-1)(1-(-1)^{b_0})2^{m-(e-1)-3}\\
		&+\sum_{e\leq l\leq m-2}l\cdot 2^{m-l-3}(1+(-1)^{b_0})+\frac{1}{2}(m-1)(1+(-1)^{b_0})\\
		=&2^{m-3}[\sum_{1\leq l\leq e-2}l\cdot 2^{-l}+(e-1)(1-(-1)^{b_0})2^{-e+1}\\&+\sum_{e\leq l\leq m-1}(1+(-1)^{b_0})l\cdot2^{-l}]+\frac{1}{4}(m-1)(1+(-1)^{b_0})
	\end{align*}
	$(\Rmnum{1}.1)$ If $b_0=0$, then
	\begin{align*}
		&\sum_{l=1}^{m-1}lN^{(\tau)}(l)\\
		=&2^{m-3}[\sum_{1\leq l\leq e-2}l\cdot 2^{-l}+\sum_{e\leq l\leq m-1}l\cdot2^{-(l-1)}]+\frac{1}{2}(m-1)\\
		=&2^{m-3}[\sum_{1\leq l\leq e-2}l\cdot 2^{-l}+\sum_{e-1\leq l\leq m-2}(l+1)\cdot2^{-l}]+\frac{1}{2}(m-1)\\
		=&2^{m-3}[\sum_{l=1}^{m-2}l\cdot2^{-l}+\sum_{l=e-1}^{m-2}2^{-l}]+\frac{1}{2}(m-1).
	\end{align*}
	Now we use the identity $\sum\limits_{l=1}^{d}l\cdot 2^{-l}=2-\frac{d+2}{2^d}\ (d\geq1)$ which can be proved by induction. We obtain
	$$\sum_{l=1}^{m-1}lN^{(\tau)}(l)=2^{m-3}[2-\frac{m}{2^m-2}+\frac{1}{2^{e-2}}-\frac{1}{2^{m-2}}]+\frac{1}{2}(m-1)=2^{m-2}+2^{m-e-1}-1.$$
	
	$(\Rmnum{1}.2)$ If $b_0=1$, then
	$$\sum_{l=1}^{m-1}lN^{(\tau)}(l)
	=2^{m-3}[\sum_{1\leq l\leq e-2}l\cdot 2^{-l}+(e-1)2^{-e+2}].$$
	For $e\geq3$,
	$$\sum_{l=1}^{m-1}lN^{(\tau)}(l)
	=2^{m-3}[2-\frac{e}{2^{e-2}}+\frac{e-1}{2^{e-2}}]=2^{m-3}[2-\frac{1}{2^{e-2}}]=2^{m-2}+2^{m-e-1}.$$
	For $e=2$, there is no $l$ satisfing $1\leq l\leq e-2=0$, $\sum\limits_{l=1}^{m-1}lN^{(\tau)}(l)
	=2^{m-3}[0+2^0]=2^{m-3}=2^{m-2}-2^{m-e-1}\ (\text{for}\ e=2).$
	For $e=1$, $\sum\limits_{l=1}^{m-1}lN^{(\tau)}(l)
	=2^{m-3}[0+0]=0=2^{m-2}-2^{m-e-1}\ (\text{for}\ e=1).$
	
	$(\Rmnum{2})$ By Lemma \ref{le3.1}, for $1\leq\tau\leq n-1$,
	$$A_s(\tau)=n-2g(\tau)=n-2(2^{m-2}+\sum_{l=1}^{m-1}lN^{(\tau)}(l)).$$
	For $b_0=0$, $A_s(\tau)=2^m-1-2(2^{m-2}+2^{m-2}+2^{m-e-1}-1)=-(2^{m-e}-1).$ For $b_0=1$, $A_s(\tau)=2^m-1-2(2^{m-2}+2^{m-2}-2^{m-e-1})=2^{m-e}-1.$ Therefore
	$$A_s(\tau)\in\{\pm(2^{m-e}-1)\arrowvert1\leq e\leq m-1\}=\{\pm(2^k-1)\arrowvert1\leq k\leq m-1\}.$$
	Moreover,
	\begin{align*}
		&\{(1+\pi^{\tau})^{-1}\arrowvert 1\leq \tau\leq n-1\}=\mathbb{F}_{2^m}\setminus\{0,1\}\\
		=&\{b_0+b_1\pi+\cdots+b_{m-1}\pi^{m-1}\neq 0,\ 1\arrowvert b_0,\ldots,b_{m-1}\in\mathbb{F}_2\}\\
		=&\{b_0+b_1\pi+\cdots+b_{e-1}\pi^{e-1}+\pi^e\arrowvert b_0,\ldots,b_{e-1}\in\mathbb{F}_2,\ 1\leq e\leq m-1\}.
	\end{align*}
	From part $(\Rmnum{1})$ we know that for $1\leq \tau\leq n-1$ and $1\leq k\leq m-1$,
	$$A_s(\tau)=2^k-1\Leftrightarrow (1+\pi^{\tau})^{-1}=1+b_1\pi+\cdots+b_{e-1}\pi^{e-1}+\pi^e,\ e=m-k.$$
	Therefore
	\begin{align*}
		&\sharp\{1\leq \tau\leq n-1\arrowvert A_s(\tau)=2^k-1\}\\
		=&\sharp\{(b_1,\ldots,b_{e-1})\arrowvert b_1,\ldots,b_{e-1}\in\mathbb{F}_2\}\\
		=&2^{e-1}=2^{m-k-1}.
	\end{align*}
	Similarly, $$A_s(\tau)=-(2^k-1)\Leftrightarrow (1+\pi^{\tau})^{-1}=0+b_1\pi+\cdots+b_{e-1}\pi^{e-1}+\pi^e$$ and
	$$\sharp\{1\leq \tau\leq n-1\arrowvert A_s(\tau)=-(2^k-1)\}=2^{e-1}=2^{m-k-1}.$$
	This completes Theorem \ref{th2}.
\end{proof}
\begin{remark}
	From  Theorem \ref{th2} we know that the conjecture, stated in Section \ref{sec1}, is true. The conjecture is concern on the distribution on absolute values $\lvert A_s(\tau)\rvert\ (1\leq\tau\leq n-1)$. In fact, from Theorem \ref{th2} and the proof we get more informations as following.
	
	$(\Rmnum{1})$.\ $\lvert A_s(\tau)\rvert=2^k-1$ if and only if $(1+\pi^{\tau})^{-1}=b_0+b_1\pi+\cdots+b_{m-k-1}\pi^{m-k-1}+\pi^{m-k}\ (1\leq k\leq m-1).$
	
	$(\Rmnum{2})$.\ $A_s(\tau)=2^k-1\Leftrightarrow b_0=1$ and $A_s(\tau)=-(2^k-1)\Leftrightarrow b_0=0$.
	
	$(\Rmnum{3})$. Both of multiplicity of $\tau\ (1\leq\tau\leq n-1)$ satisfying $A_s(\tau)=2^k-1$ and $A_s(\tau)=-(2^k-1)$ are the same number $2^{m-k-1}$.
\end{remark}

\section{Conflicts of Interest:}
The authors declare that they have no conflicts of interest to report regarding the present study.


\begin{thebibliography}{9}
	
	\bibitem{1C1} Z. Chen, Z. Niu, Y. Sang, C. Wu, Arithmetic autocorrelation of binary {\boldmath$m$}-sequences, Cryptologia, DOI: 10.1080/01611194.2022.207116, 2022.
	
	\bibitem{2C2} Z. Chen, Z. Niu, A. Winterhof, Arithmetic crosscorrelation of pseudorandom binary sequences of coprime periods, IEEE Trans. on Inform. Theory, 68(11) (2022) 7538-7544.
	
	\bibitem{3G1} M. Goresky, A.Klapper, Arithmetic crosscorrelations of feedback with carry shift register sequences, IEEE Trans. on Inform. Theory, 43(4) (1997), 1342-1345.
	
	\bibitem{4G2} M. Goresky, A.Klapper, Some results on the arithmetic correlation of sequences. In $\langle$Sequences and Their Applications-SETA 2008$\rangle$, Lecture Notes in Computer Science, vol.5203, (2008), 71-80.
	
	\bibitem{5G3} M. Goresky, A.Klapper, Statistical properties of the arithmetic correlation of sequences, International Jour. of Foundations of Computer Science, 22(6) (2011), 1297-1315.
	
	\bibitem{6G4} M. Goresky, A.Klapper, $\langle$Algebraic Shift Register Sequences$\rangle$, Combiridege University Press, 2012.
	
	\bibitem{7H1} R. Hofer,  A. Winterhof, On the arithmetic autocorrelation of the Legendre sequence, Advances in Mathematics of Communications, 11(1)(2017), 237-244.
	
	\bibitem{8H2} R. Hofer,  L. Merai, A. Winterhof, Measures of pseudorandomness arithmetic autocorrelation and correlation measure, In $\langle$Number theory-Diophantine probiems, uniform distribution and applications$\rangle$, C. Elsholtz and P. Grabner edit. Springer, 303-312, 2017.
	
	\bibitem{9Mand} D. M. Manddelbaum, Arithmetic codes with large distance, IEEE Trans. on Inform. Theory, 13(2) (1967), 237-242.
	
	\bibitem{10W} H. Wang, Q. Wen and J. Zhang, GLS: New class of generalized Legendre sequences with optimal arithmetic crosscorrelation, RAIRO-Theor. Inform. Appl., 47(2013), 371-388.
\end{thebibliography}
\end{document}